\documentclass[12pt,notitlepage,reqno]{amsart}

\usepackage{amsfonts}
\usepackage{amsthm}
\usepackage{amstext}
\usepackage{mathtools}

\newcommand{\supp}{\operatorname{supp}}

\newcommand{\R}{{\mathbb R}}

\newcommand{\N}{{\mathbb N}}

\theoremstyle{plain}
\newtheorem{thm}{Theorem}[section]
\newtheorem{prop}[thm]{Proposition}
\newtheorem{lemma}[thm]{Lemma}
\theoremstyle{definition}

\newtheorem{assump}[thm]{Assumption}
\newtheorem*{acknowledgement}{Acknowledgement} 
   
\newenvironment{pf}{\par\medskip\noindent\textit{Proof}:\,}{\hspace*{\fill}\qed\medskip\par\noindent} 
\newenvironment{pf*}[1]{\par\medskip\noindent\textit{#1}\,:}{\hspace*{\fill}\qed\medskip\par\noindent}   
  
\newtheorem{remark}[thm]{Remark}

\numberwithin{equation}{section}

\title[Pseudorelativistic electron densities]{The electron densities
  of pseudorelativistic eigenfunctions are smooth away from the nuclei} 

\thanks{\copyright\ 2008 by the authors. This article may be
  reproduced in its entirety for non-commercial purposes.}

\author[S. Fournais and T. \O stergaard S\o rensen]{S\o ren Fournais
        and Thomas \O stergaard S\o rensen}

\address[S. Fournais]
        {Department of Mathematical Sciences, 
         University of Aarhus, 
         Ny Munkegade, Building 1530, 
         DK-8000 \AA rhus C, Denmark.} 
\email{fournais@imf.au.dk}           
\address[S. Fournais on leave from]
        {CNRS and Laboratoire de
         Math\'{e}matiques d'Orsay, 
         Univ Paris-Sud, 
         Orsay CEDEX, F-91405, France.} 

\address[T. \O stergaard S\o rensen]
        {Department of
         Mathematical Sciences,
         Aalborg University,
         Fredrik Bajers Vej 7G,
         DK-9220 Aalborg East,
         Denmark.}
\email{sorensen@math.aau.dk}

\begin{document}
\thispagestyle{empty}
\date{\today}
\begin{abstract} We consider a pseudorelativistic model of atoms and
  molecules, where the kinetic energy of the electrons is given by
  $\sqrt{p^2+m^2}-m$. In this model the eigenfunctions are generally
  not even bounded, however, we prove that the corresponding
  one-electron densities are smooth away from the nuclei.
\end{abstract}
\maketitle
\section{Introduction and results}
It was proved recently \cite{SecondWien,Taxco} that the one-electron
densities of atomic and molecular eigenstates are smooth away from the
nuclei (actually, real analyticity was proved in \cite{Analytic}). The
model studied was the non-relativistic Schr\"{o}dinger operator with
fixed nuclei. The proofs in \cite{SecondWien,Taxco} depend heavily on
special properties of the non-relativistic kinetic energy operator
$-\Delta$. However, the strategy of large parts of the proof is very
robust. In the present paper we generalise the result to the case of
so-called pseudorelativistic molecules. 

We consider an $N$-electron molecule with $L$ fixed nuclei. The
pseudo\-re\-la\-tivi\-stic Hamiltonian is (in units where
\(\hbar=c=1\)) given by 
\begin{align}
   \label{Hmol}
   {\mathbf H}_{N,L}(\mathbf R,\mathbf
   Z)=\sum_{j=1}^N
   \Big\{T(p_j)-\sum_{\ell=1}^L
   \frac{Z_{\ell} \alpha }{|x_j-R_{\ell}|}\,\Big\}
   +\sum_{1\le i<j\le N}\frac{\alpha}{|x_i-x_j|}\,,
\end{align}
where the kinetic energy $T(p_j)$ of the $j$'th electron is given by
the operator
$$
  T(p) = \sqrt{p^2+m^2}-m = \sqrt{-\Delta+m^2}-m\,,
$$
with $m\in [0,\infty)$ being the mass of the electron; \(\alpha\) is
the fine structure constant (in these units, \(\alpha=e^2\), with
\(e\) the unit charge). In \eqref{Hmol}, $\mathbf R=(R_1,R_2,\dots
,R_L)\in\mathbb R^{3L}$, \(R_{\ell}\neq R_{k}\) for \(k\neq \ell\),
denote the positions of the $L$ nuclei whose positive charges are
given by $\mathbf Z=(Z_1, Z_2,\dots,Z_L)$. The positions of the $N$
electrons are denoted by ${\bf x}= (x_1,x_2,\dots,x_N)\in\mathbb
R^{3N}$ where $x_j$ denotes the position of the $j$'th electron in
$\mathbb R^3$; \(\Delta_j\) is the Laplacian with respect to
\(x_j\). We write \(\nabla=(\nabla_1,\ldots,\nabla_N)\) for the
gradient operator in \(\R^{3N}\). In \eqref{Hmol} we have omitted the
nucleus-nucleus interaction, $\sum_{\ell<k}\frac{Z_{\ell}Z_k
  \alpha}{|R_{\ell}-R_k|}$, since this is just an additive constant. 

The natural space for studying the operator \({\mathbf
  H}_{N,L}(\mathbf R,\mathbf Z)\) is, in view of the Pauli Exclusion
Principle, the antisymmetric spinor space,\par\noindent
$\wedge_{j=1}^N L^2({\mathbb R}^3; {\mathbb C}^2)$, however, our
results will not depend on spin and we do therefore not impose this
antisymmetry condition. Instead we work on the space $L^2({\mathbb
  R}^{3N})$.

We will assume that $0 < Z_{\ell} \alpha < 2/\pi$ for all $\ell \in
\{1,\ldots,L\}$.\footnote{The experimental value of the fine structure
  constant is $\alpha \approx 1/137$. For this value of $\alpha$, 
  $2/(\pi \alpha) \approx 87$} In this case we get from
\cite[Proposition 2.2]{DauLi83} 
(see also \cite{He77} and \cite{We75} for the case of Hydrogen)
that the negative Coulomb potentials constitute a {\it small form
perturbation} of the (total) kinetic energy (i.e., it is relatively
form bounded with relative bound less than one). The electron-electron
interactions being positive, and relative form bounded too, we get
that the quadratic form
\begin{align}
  \label{eq:formDef}
  \mathfrak{q}(u,v)&:=\Big\langle u\ , \sum_{j=1}^N
   T(p_j)\,v \Big\rangle-\Big\langle u\ , \sum_{j=1}^N \sum_{\ell=1}^L
  \frac{Z_{\ell} \alpha }{|x_j-R_{\ell}|}\,v \Big\rangle  
  \\&\quad
  +\Big\langle u,\sum_{1\le i<j\le N}\frac{\alpha}{|x_i-x_j|}\,
  v\Big\rangle 
  \ , \quad u,v \in H^{1/2}(\R^{3N})\,,
  \nonumber
\end{align}
is closed and semi-bounded. Here, \(\langle\cdot,\cdot\rangle\) is the
scalar product in \(L^2(\R^{3N})\). Hence, we can define the operator
${\mathbf H} \equiv {\mathbf H}_{N,L}(\mathbf R,\mathbf Z)$ as the
corresponding (unique) self-adjoint operator. It satisfies 
\[ H^1(\R^{3N}) \subset {\mathcal{D}({\bf H})\subset
  H^{1/2}(\R^{3N})}\,,\] 
and
\begin{align}
  \label{eq:form-op}
  \mathfrak{q}(u,v)=\langle u,{\bf H}v\rangle\,, \quad v\in
  \mathcal{D}({\bf H})\,,\quad u\in H^{1/2}(\R^{3N})\,.
\end{align}
Here, \(\mathcal{D}(\bf H)\) denotes the operator domain of \({\bf
  H}\); we denote its form domain by \(\mathcal{Q}(\bf H)\). 
All this follows from (the statements and proofs of)
\cite[Theorem~X.17]{R&S2} and \cite[Theorem~VIII.15]{R&S1}.
See \cite{LiYau88} for further references on 
${\mathbf H}_{N,L}(\mathbf R,\mathbf Z)$.

Suppose $\psi \in L^2({\mathbb R}^{3N})$ is an eigenfunction of
${\mathbf H}$, i.e., there exists $E \in {\mathbb R}$ such that 
$$
  {\mathbf H} \psi = E\psi\,.
$$
We define the {\it one-electron density}
$\rho\in L^1(\R^3)$ (associated to $\psi$) by 
\begin{align}\label{def:rho}
  \rho(x)&=\sum_{j=1}^N \rho_j(x)\nonumber
  \\&=\sum_{j=1}^N\int_{\mathbb
    R^{3N}}|\psi(x_1,\ldots,x_N)|^2\,\delta(x-x_j)\,dx_1\cdots dx_N\,.
\end{align}

The main result of this paper is the following.
\begin{thm}
\label{thm:density_smooth}
Let \(\psi \in L^2({\mathbb R}^{3N}) \) be an eigenfunction of
\({\mathbf H}\). Let the associated density \(\rho\) be as defined in  
\eqref{def:rho}.

Then
\begin{align}
  \label{eq:smooth}
  \rho&\in C^{\infty}\big({\mathbb R}^{3}\setminus
  \{R_1,\ldots, R_L\}\big)\,.
\end{align}
\end{thm}

\begin{remark}\label{rem:1}
  \(\, \)
\begin{enumerate}
\item[\rm (i)]
Theorem~\ref{thm:density_smooth} will follow from the more
general abstract Theorem \ref{thm:abstract} below.
\item[\rm (ii)]
We state Theorem~\ref{thm:density_smooth} for Coulomb
interactions, but it holds for more general potentials.
For instance, one can use the Yukawa potential $\frac{e^{-c
|x|}}{|x|}$, with $c>0$, in one or all of the two-particle
interactions. See Theorem~\ref{thm:abstract} below for a more general
statement of the result.
\item[\rm (iii)]
Since we are only interested in regularity
properties of $\rho$, we can study each of the
(finitely many) terms in \eqref{def:rho} separately. 
We will restrict ourselves to proving the statement in
\eqref{eq:smooth} for  
\begin{align}
  \label{eq:rho}
  \rho_1(x) := \int_{\R^{3N-3}} |\psi(x,x_2,\ldots, x_N)|^2 \, dx_2\cdots
  dx_N\,,
\end{align}
the proof for the other terms being analogous.
Furthermore, to simplify the presentation, we limit ourselves to the
atomic case (\(L=1, R_1=0, Z_1=Z, 0<Z\alpha<2/\pi\)). 
\end{enumerate}
\end{remark}

\noindent{\bf Notation.}
We denote by ${\mathcal B}^{\infty}(U)$ the smooth functions with
bounded derivatives on the open set $U$, i.e., 
$$
  {\mathcal B}^{\infty}(U) = \big\{ u \in C^{\infty}(U) \,\big|
  \,\partial^\alpha u \in L^{\infty}(U)\text{ for all } \alpha\,\big\}\,. 
$$

\section{The abstract theorem}
Our main interest in this paper is the regularity of one-electron
densities of pseudorelativistic atoms and molecules with Coulomb
interactions, as stated in Theorem~\ref{thm:density_smooth}. However,
our result holds in a more general case, which we will state
here. 

It is known that, in the case of relativistic atoms, the potential
energy is not a small {\it operator} perturbation of the kinetic
energy, if the values of $\alpha, N$, and $Z$ become too large. (This
is also the case in other relativistic models than the one studied
here.) In this case, as discussed in the introduction,  the
Hamiltonian is only defined as the (unique) self-adjoint operator
associated to a semi-bounded closed quadratic form. On the other hand,
the pseudorelativistic kinetic energy has an extra, important
property: It is the generator of a positivity preserving semigroup. 

Our abstract conditions below are thus based on the kinetic energy 
\({\bf T}\) below being the generator of a positivity preserving
semigroup. This fact follows from the explicit formula for the
integral kernel of the semigroup generated by \(T(p)\); see
e.g. \cite[7.11(11)]{Li-Loss}.  

The Hamiltonians considered will be of the form
\begin{align}
  {\mathbf H} = {\mathbf T} + {\bf V} \,,
\end{align}
where (with \({\bf p}=(p_1,\ldots,p_N)\in\R^{3N}\))
\begin{align}
  \label{def:T}
  {\mathbf T} &= {\bf T}({\bf p})=\sum_{j=1}^N T(p_j)
  = \sum_{j=1}^N \sqrt{-\Delta_j+m^2}-m\,,\\
  {\bf V} &= {\bf V}({\mathbf x}) = \sum_{j=1}^{N} V_j(x_j) +
  \sum_{1\leq j< k\leq N} 
  W_{j,k}(x_j-x_k)\,.
\end{align}

The following are the assumptions on the potential \({\bf V}\). 
\begin{assump}
\label{as:form}
\begin{itemize}
\item[(i)] 
\begin{itemize}
\item
For all \(j\in\{1,\ldots,N\}\), 
$$
  V_j\in C^{\infty}\big({\mathbb R}^3\setminus \{0\}\big)\cap
  {\mathcal B}^{\infty}\big({\mathbb R}^3\setminus
  B(0,1)\big)\,.
$$
\item
For all $Q \subset \{1,\ldots, N\}$, the quadratic form on
  $\otimes_{j\in Q} L^2({\mathbb R}^3)$  given by the multiplication
  operator  
$
  {\bf V}_Q:= \sum_{j\in Q} V_j(x_j)
$
is a small form perturbation of 
$
  {\mathbf T}_{Q}:=\sum_{j\in Q} |p_j|\,.
$
\end{itemize}
\item[(ii)] 
For all \(j,k\in\{1,\ldots,N\}\) with \(j\neq k\), 
 \begin{itemize}
 \item $W_{j,k} \geq 0$ pointwise and \(W_{j,k}(x)=W_{k,j}(-x)\).
 \item $W_{j,k} \in C^{\infty}\big({\mathbb R}^3\setminus \{0\}\big)\cap
   {\mathcal B}^{\infty}\big({\mathbb R}^3\setminus
   B(0,1)\big)\,.
 $
 \item Multiplication by $W_{j,k}$ defines a bounded operator from
   $H^1({\mathbb R}^3)$ to $L^2({\mathbb R}^3)$ (by
   interpolation boundedness from \par\indent$H^{1/2}(\R^3)$ to
   $H^{-1/2}(\R^3)$ therefore follows). 
 \end{itemize}
\end{itemize}
\end{assump}
Under the above assumptions it is clear that ${\mathbf H}={\mathbf T}
+ {\bf V}$ is well defined as the (unique) self-adjoint operator of
the the corresponding closed and semi-bounded quadratic form (see the
introduction for details).

The main abstract result of this paper is the following.
\begin{thm}
\label{thm:abstract}
Let $m\geq 0$ and let ${\mathbf T}$ be the (total) pseudorelativistic
kinetic energy operator 
\begin{align}
  {\mathbf T} = \sum_{j=1}^N \sqrt{-\Delta_j+m^2}-m\,.
\end{align}
Let functions 
\begin{align}
  V_j &: {\mathbb R}^3 \rightarrow {\mathbb R} ,\quad j \in \{1, \ldots,
  N\},\nonumber\\ 
  W_{j,k} &: {\mathbb R}^3 \rightarrow {\mathbb R} ,\quad j,k \in \{1,
  \ldots, N\}, j\neq k\,, \nonumber 
\end{align}
be given such that Assumption~\ref{as:form} is satisfied, and let
\begin{align*}
  {\bf V}({\mathbf x}) = \sum_{j=1}^{N} V_j(x_j) + \sum_{1\leq j< k\leq N}
  W_{j,k}(x_j-x_k)\,.
\end{align*}
Let ${\mathbf H} = {\mathbf T} + {\bf V}$ be the self-adjoint operator
associated to the corresponding quadratic form (closed on
$H^{1/2}({\mathbb R}^{3N})$).
Let finally $\psi\in L^2({\mathbb R}^{3N})$ be an eigenfunction of
${\mathbf H}$ and let $\rho$ be the associated
density as defined in \eqref{def:rho}. 

Then
\begin{align}
  \rho \in C^{\infty}\big({\mathbb R}^3 \setminus \{0\}\big)\,.
\end{align}
\end{thm}

\begin{remark}
As pointed out in Remark~\ref{rem:1} (i),
Theorem~\ref{thm:density_smooth} follows from Theorem
\ref{thm:abstract}. 
\end{remark}

\begin{proof}[Proof of Theorem~\ref{thm:abstract}]$\,$\\
The smoothness of $\rho$ is a direct consequence of
Propostion~\ref{lem:diff_parallel} below. The argument is exactly the
same as the one given in \cite[Section 3]{Taxco} in the proof of
\cite[Theorem 1.1]{Taxco}. We therefore omit the details. 
\end{proof}
All that remains is to (state and) prove
Proposition~\ref{lem:diff_parallel} below. 

\section{The parallel differentiation}
\label{parallel}
The fact that one is allowed to differentiate the eigenfunction
\(\psi\) parallel to the singularities of the (total) potential \(V\)
is the key ingredient in proving the smoothness of the density
\(\rho\). This approach was carried out for the non-relativistic
Schr\"{o}dinger operator---that is, with \(T(p_j)=-\Delta_j\) in
\eqref{Hmol}---in \cite[Proposition 1]{SecondWien} (see also
\cite{Taxco}). We sketch the main ideas before giving the exact
statement of the result (Proposition~\ref{lem:diff_parallel} below)
and its proof. 

Let $u \in L^2(\R^d)$ and ${\mathcal V}\in L^{\infty}(\R^d)$, and
assume that  
\begin{align}\label{eq:EQN}
  \Delta u = {\mathcal V}u\,.
\end{align}
Then \eqref{eq:EQN} implies that $u \in H^2(\R^d)$, in
particular, $\partial u \in L^2(\R^d)$ for any derivative \(\partial\).
Assume furthermore that, for some specific directional derivative
$\partial_{a} = \sum_j a_j \partial_j$, $a_j \in \R$, we have
$\partial_{a} {\mathcal V} \in L^{\infty}(\R^d)$. As just argued,
\(\partial_{a}u\in L^2(\R^d)\). Then, by
differentiation of \eqref{eq:EQN}, we find that
\begin{align}\label{eq:EQN2}
  \Delta(\partial_{a}u)={\mathcal V}\partial_{a} u
  +(\partial_{a}{\mathcal V}) u\,, 
\end{align}
from which it follows that, in fact, $\partial_{a} u \in H^2(\R^d)$. 
Moreover, the above argument is easily localised: If $\partial_{a}
{\mathcal V} \in L^{\infty}(U)$ for some open set $U \subset \R^d$,
then we can conclude that $\partial_{a} u \in H^2(U)$. 

Using this idea (and an induction argument) on the eigenvalue
equation one finds that eigenfunctions of the non-relativistic
molecular Hamiltonian are smooth in certain directions and on certain
open sets (see Proposition~\ref{lem:diff_parallel} for a precision of
the geometry, which is the same as in the non-relativistic case). In
the molecular case the (Coulomb) potential is not a bounded function,
but one easily sees that the argument carries over to the case of
potentials \(\mathcal{V}\) which are a small {\it operator}
perturbation of the kinetic energy.  

For the pseudorelativistic operator in \eqref{Hmol} this procedure
does not work immediately, since we cannot separate the kinetic and
potential energies: Since the potential ${\bf V}$ is only a small
quadratic form perturbation of the kinetic energy \({\bf T}\), the
operator \({\bf H}={\bf T}+{\bf V}\) is only given as a form sum. 

The idea is then to move the term \(\mathcal{V}u\) to the left hand
side in \eqref{eq:EQN} to find the following substitute for the
argument above. Let the operator ${\mathfrak H}$ be self-adjoint with
operator domain (contained in) $H^s(\R^d)$, for some $s\geq
1$. Suppose $u\in L^2(\R^d)$ satisfies (in the weak sense) the
equation 
\begin{align}\label{eq:Stupid2}
  {\mathfrak H} u = v \in L^2(\R^{d})\,.
\end{align}
It follows that $u \in {\mathcal D}({\mathfrak H}) \subset
H^s(\R^d)$. If furthermore $v \in H^1(\R^d)$ one can then take a
derivative in \eqref{eq:Stupid2} and use arguments as above to
conclude that \(\partial u\in{\mathcal D}({\mathfrak H}) \subset
H^s(\R^d)$.

However, in our case it is not easy to identify the operator domain of
${\mathbf H}$. By the definition as a form sum, we only get that
$H^{1}(\R^{3N}) \subset  {\mathcal D}({\mathbf H}) \subset
H^{1/2}(\R^{3N})$. That is, we cannot take one derivative on something
in \(\mathcal{D}(\mathbf{H})\) as explained above and still be sure to
obtain a function in $L^2(\R^{3N})$. Furthermore, the relativistic
kinetic energy is not {\it local}, so introduction of cut-off
functions in the induction argument becomes somewhat more complicated.  

Nevertheless, the above idea of a proof and therefore the main
technical step in \cite{Taxco}---Proposition~\ref{lem:diff_parallel}
below---can still be justified. That is, the strategy of repeatedly
differentiating an equation of the form \eqref{eq:Stupid2} in `good'
directions remains: 
We partially identify the operator domain ${\mathcal D}({\mathbf H})$
in order to be able to take one {\it parallel} derivative
(\(\partial_{x_P}\) below) on functions therein.
\begin{prop}
\label{lem:diff_parallel}
Let $P,Q$ be a partition of $\{1,\ldots,N\}$ satisfying
$$
  P \neq \emptyset, \,\,\,\,\,\,\,\,\,\,\,\,
  P\cap Q = \emptyset, \,\,\,\,\,\,\,\,\,\,\,\,
  P \cup Q =\{1,\ldots,N\}\,.
$$
Define, for $P,Q$ as above and $\epsilon > 0$,
\begin{align}\label{eq:def_U_P}
  U_P(\epsilon) = \big\{ (x_1,\ldots,x_N) \in {\mathbb
  R}^{3N}\,\big|&\, |x_j| >
  \epsilon \mbox{ for } j \in P, \nonumber \\
  & |x_j-x_k| > \epsilon \mbox{ for }
  j \in P, k \in Q \big\}\,.
\end{align}
Define also
\begin{align}\label{def:xP}
  x_P = \frac{1}{\sqrt{|P|}} \sum_{j\in P} x_j\quad (\in\R^3)\,.
\end{align}
Let furthermore \({\bf H}\) be as in Theorem~\ref{thm:abstract}, and
let $\psi \in L^2({\mathbb R}^{3N})$ be an eigenfunction of 
${\mathbf H}$, i.e., there exists $E \in {\mathbb R}$ such that
$$
  {\mathbf H} \psi = E \psi\,.
$$

Then
$$
  \partial^{\gamma}_{x_P} \psi \in L^2(U_P(\epsilon))
  \text{ for all }
  \gamma \in {\mathbb N}^3\,.
$$
\end{prop}

\begin{proof}
Since the proof is somewhat technical we split it in a
number of steps in order to make the structure more
transparent. We first prove a lemma on localization.
\begin{lemma}\label{lem:localize}
Let \(\varphi\in\mathcal{B}^{\infty}(\R^{3N})\) and \(u\in
\mathcal{D}(\mathbf{H})\). Then \(\varphi
u\in\mathcal{D}(\mathbf{H})\) and
\begin{align}\label{eq:com}
  \mathbf{H}(\varphi u)=\varphi(\mathbf{H}u)+Bu\,,
\end{align}
where \(B\in\mathcal{B}(L^2(\R^{3N}))\) is the commutator
\([{\mathbf{T}},\varphi]\). 
\end{lemma}
\begin{proof}
Notice first that $\varphi u \in {\mathcal Q}({\mathbf H})$ since
$u\in{\mathcal D}({\mathbf H})\subset{\mathcal Q}({\mathbf
  H})=H^{1/2}(\R^{3N})$ and multiplication by $\varphi$ maps
$H^s(\R^{3N})$ into itself for all $s\in{\mathbb R}$. Let
$v\in{\mathcal D}({\mathbf H})\subset{\mathcal Q}({\mathbf H})$,
then also \(\overline{\varphi}v\in{\mathcal Q}(H)\), and, since $u\in
{\mathcal D}({\mathbf H})$ and \(\mathfrak{q}\) is symmetric (see
\eqref{eq:formDef}), 
\begin{align}
  \label{eq:1}
  \mathfrak{q}(\varphi u, v)= 
  \langle\varphi u,{\mathbf H}v\rangle\ ,\quad
  \mathfrak{q}(u,\overline{\varphi}v)=\langle{\mathbf
    H}u,\overline{\varphi}v\rangle\,. 
\end{align}
Now, we can calculate on a form core ($C^{\infty}_0(\R^{3N})$) to obtain
\begin{align}\label{eq:2}
  \mathfrak{q}(\varphi u, v) = \mathfrak{q}(u, {\overline{\varphi}} v)
  + \langle B u, v\rangle\,, 
\end{align}
where $B$ is the operator $[{\mathbf T}, \varphi]$, which  is bounded on
$L^2({\mathbb R}^{3N})$ since $\varphi \in {\mathcal B}^{\infty}({\mathbb
  R}^{3N})$ (see Lemma~\ref{lem:PSDO} below). It follows
from \eqref{eq:1} and \eqref{eq:2} that 
\begin{align}
  \label{eq:commute}
  \langle\varphi u,{\mathbf H}v\rangle
  =\langle\varphi{\mathbf H}u + B u ,v\rangle
  \ \text{ for all }\ 
  v\in \mathcal{D}({\mathbf H})\,.
\end{align}
Since $\varphi {\mathbf H} u + B u \in L^2(\R^{3N})$ and
\(\mathcal{D}({\mathbf H})\) is dense in \(L^2(\R^{3N})\) we deduce from
\eqref{eq:commute} that $\varphi u\in{\mathcal D}({\mathbf
  H^*})={\mathcal D}({\mathbf H})$ and that \eqref{eq:com} holds. 
This proves the lemma.
\end{proof}

\noindent{\bf An auxiliary operator.}
We introduce the following two operators:
\begin{align}\label{def:H-Q}
  {\mathbf H}_Q & = \sum_{j \in Q} \big( T(p_j) + V_j(x_j) \big)
  + \sum_{j,k \in Q,j<k} W_{j,k}(x_j-x_k)\,,
\end{align} 
on $\otimes_{j\in Q} L^2({\mathbb R}^3)$, 
and
\begin{align}\label{def:H-P}
  {\mathbf H}_P & = \sum_{j \in P} T(p_j) + \sum_{j,k \in P,
  j<k} W_{j,k}(x_j-x_k)\,,
\end{align}
on $\otimes_{j\in P} L^2({\mathbb R}^3)$.

By Assumption~\ref{as:form} (notice that the $W_{j,k}$ are
non-negative, and that \(T(p_j)-|p_j|\) 
is a bounded operator on
\(L^2(\R^3)\)) the quadratic form 
defined by ${\mathbf H}_Q$ is closed and bounded from below on
\(H^{1/2}(\R^{3|Q|})\). The operator ${\mathbf H}_Q$ is
then defined as the (unique) self-adjoint operator associated to this
form; see \cite[Theorem~VIII.15]{R&S1}.

It follows  from Lemma~\ref{lem:davies_faris} in
Appendix~\ref{app:A} that ${\mathbf H}_P$ is self-adjoint with domain
$H^{1}({\mathbb R}^{3|P|})$. We here used Assumption~\ref{as:form}
(ii) and that \(T(p)\) (and therefore, \(\sum_{j \in P} T(p_j)\)), as
mentioned earlier in this section, is the generator of a positivity
preserving semigroup. 

Define furthermore
$$
  \widehat{{\mathbf H}} 
  = {\mathbf H}_Q \otimes 1 + 1 \otimes {\mathbf H}_P
$$
on
$$
  L^2({\mathbb R}^{3N}) \simeq
  \left( \otimes_{j\in Q} L^2({\mathbb R}^3) \right) \otimes
  \left( \otimes_{j\in P} L^2({\mathbb R}^3)\right)\,.
$$
Since ${\mathbf H}_Q$ and ${\mathbf H}_P$ are bounded below, it
follows from results on tensor products
  \cite[p.\ 86]{Monvel1} 
that
$\widehat{{\mathbf H}}$ is self-adjoint with domain 
\begin{align}
\label{eq:DomHat}
  {\mathcal  D}(\widehat{{\mathbf H}})
  &= \big[{\mathcal D}({\mathbf H}_Q) \otimes L^{2}({\mathbb
    R}^{3|P|})\big] 
  \cap
  \big[L^{2}({\mathbb R}^{3|Q|})\otimes {\mathcal D}({\mathbf
    H}_P)\big] \nonumber\\ 
  &\subseteq  
  L^{2}({\mathbb R}^{3|Q|})\otimes {\mathcal D}({\mathbf H}_P) 
 =  L^{2}({\mathbb R}^{3|Q|}) \otimes H^{1}({\mathbb R}^{3|P|})\,.
\end{align}
Choose $\widehat{V}_j \in {\mathcal B}^{\infty}({\mathbb R}^3)$ for $j
\in P$ and $\widehat{W}_{j,k} \in {\mathcal B}^{\infty}({\mathbb
  R}^3)$ for $j \in 
P, k
\in Q$ (and $k \in P, j \in Q$) satisfying
\begin{align*}
  \widehat{V}_j = V_j \text{ on } {\mathbb R}^3\setminus
  B(0,\epsilon/2) \ \ \text{ and }\ \  
  \widehat{W}_{j,k} = W_{j,k} \text{ on } {\mathbb R}^3\setminus
  B(0,\epsilon/2)\,. 
\end{align*}
This is possible by Assumption~\ref{as:form}.
Define finally
\begin{align}
\label{eq:IP}
  \widetilde{{\mathbf H}} &= \widehat{{\mathbf H}} + I_P\,,\\
  I_P({\mathbf x})  &= \sum_{j\in P}
  \widehat{V}_j(x_j) + \sum_{(j\in P,k\in Q)\cup(j\in Q,k\in P)}
  \!\!\!\!\!\!\!\!\!\!\!\!
  \widehat{W}_{j,k}(x_j-x_k)\,.\nonumber
\end{align}
The operator \(\widetilde{\bf H}\) is self-adjoint, with \( {\mathcal
  D}(\widetilde{{\mathbf H}}) = {\mathcal D}(\widehat{{\mathbf H}})\),
since $\widehat{V}_j, \widehat{W}_{j,k}\in L^{\infty}(\R^3)$. 
We have (in the form sense)
\begin{align}\label{eq:TildeV}
  \widetilde{{\mathbf H}} = {\mathbf T} + \widetilde{{\mathbf V}}
\end{align}
with
\begin{align}\label{def:tildeV}
  \widetilde{\bf V}({\bf x}) = I_P({\bf x})
  + \sum_{j \in Q}V_j(x_j)
  +\!\!\!\!\!
  \sum_{(j,k \in P,
  j<k)\cup(j,k \in Q,j<k)} 
  \!\!\!\!\!\!\!\!\!\!\!\!\!\!\!
  W_{j,k}(x_j-x_k)\,.
\end{align}

Let $\widetilde{\mathfrak{q}}$ be the quadratic form
associated with $\widetilde{{\mathbf H}}$. An approximation
argument, using that
$C^{\infty}_0(\R^{3N})$ is a form core for both \(\mathfrak{q}\) and
\(\widetilde{\mathfrak{q}}\), gives that  
for $u,v \in H^{1/2}({\mathbb R}^{3N})$ with
$\supp u \subset U_P(\epsilon/2)$,
\begin{align}\label{eq:formsEqual}
  \mathfrak{q}(u,v) = 
  \widetilde{\mathfrak{q}}(u,v)\,. 
\end{align}

\ %

\noindent{\bf The parallel differentiation.}
Let $f_1,f_2\in C^{\infty}({\mathbb R})$ be a partition of unity on
${\mathbb R}$ satisfying that $f_1$ is non-increasing and $f_1(t)=1$
for $t\leq 5/4$, $f_1(t)=0$ for $t\geq2$, $f_1+f_2=1$. 

For $\epsilon>0$ and $P \subset \{1,\ldots,N\}$, $P\neq \emptyset$ define 
\begin{align}
  \varphi_{P,\epsilon}({\bf x}) := \prod_{j\in P} f_2(2|x_j|/\epsilon)
  \prod_{j\in P, k\in Q} f_2(2|x_j-x_k|/\epsilon)\,. 
\end{align}
Then $\varphi_{P,\epsilon} \in {\mathcal B}^{\infty}({\mathbb
  R}^{3N})$ and $\supp \varphi_{P,\epsilon} \subset U_P(\epsilon/2)$. 

We will prove the following lemma, by induction in
\(k\in\N\cup\{0\}\). 
Notice that part (1) in the lemma implies  that
$\partial^{\gamma}_{x_P} \psi \in L^2(U_P(\epsilon))$. Therefore,
Proposition~\ref{lem:diff_parallel} clearly follows once we have
proved 
Lemma~\ref{lem:S(k)}. 
\begin{lemma}\label{lem:S(k)}
For all \(k\in\N\cup\{0\}\) the following holds:\\
For all $\epsilon>0$, all $P\subset\{1,\ldots,N\}$ with $P \neq
\emptyset$, and all $\gamma \in  
{\mathbb N}^3$ with $|\gamma| \leq k$:
\begin{enumerate}
\item $\partial_{x_P}^{\gamma} (\varphi_{P,\epsilon} \psi) \in {\mathcal
      D}({\mathbf H})\cap {\mathcal D}(\widetilde{{\mathbf H}})$.
\item If $\gamma = \gamma_1 + \cdots + \gamma_k$, with $|\gamma_j|=1$
      for all $j$, then 
\begin{align}
\label{eq:DiffUnder}
  {\ }\qquad{\mathbf H}(\partial_{x_P}^{\gamma} (\varphi_{P,\epsilon} \psi)) &=
  E \partial_{x_P}^{\gamma} (\varphi_{P,\epsilon} \psi) +
  \partial_{x_P}^{\gamma}  [{\mathbf T},\varphi_{P,\epsilon}] \psi
  \\
  &\quad-\sum_{j=1}^k 
  \partial_{x_P}^{\gamma_1+\cdots+\gamma_{j-1}}
  \big\{(\partial_{x_P}^{\gamma_j} I_P) 
  \partial_{x_P}^{\gamma_{j+1}+ \cdots + \gamma_k}
  (\varphi_{P,\epsilon} \psi)\big\}\,.\nonumber
\end{align}
\end{enumerate}
\end{lemma}
\begin{pf}
We proceed by induction.

It follows from Lemma~\ref{lem:localize} that the statement is correct
for $k=0$ (in which case 
\eqref{eq:DiffUnder} reduces to \eqref{eq:com}, when using that
\({\mathbf H}\psi=E\psi\)). 

Suppose that the statement is true for some $k\geq0$. Let $\gamma \in
{\mathbb N}^3$ with $|\gamma| = k$, and write
$u_{\gamma}=\partial^{\gamma}_{x_P} (\varphi_{P,\epsilon} \psi)$.

Let ${\bf e}_P$ be any of the three unit vectors
in
${\mathbb R}^{3N}$ which define the directions of $x_P$. More
precisely, introduce the canonical
basis for ${\mathbb R}^{3N}$, $\{{\bf e}_{j}^{k}\}$ with $j\in \{1, \ldots,
N\}$, $k\in \{1,2,3\}$. Then the vector ${\bf e}_P$ is one of the three
possibilities 
\begin{align}\label{def:e's}
  {\bf e}_P^k := \frac{1}{\sqrt{|P|}} \sum_{j\in P} {\bf e}_j^k\ , \ k 
  \in \{1,2,3\}\,. 
\end{align}
Let  \(\partial_{{\bf e}_P}={\bf e}_P\cdot\nabla\) be the directional 
derivative in the direction \({\bf e}_P\), and define the self-adjoint
operator ${\bf e}_P\cdot {\bf p} = -i {\bf e}_P\cdot\nabla$ with
domain 
$$
  {\mathcal D}({\bf e}_P\cdot {\bf p}) = \{ f \in L^2(\R^{3N}) \,\big|\,
  \partial_{{\bf e}_P}f \in L^2(\R^{3N})\}\,. 
$$
Let furthermore, for 
$t\in \R$, \(\tau_{t{\bf e}_P}\) be the translation operator
$(\tau_{t{\bf e}_P}f)({\bf x}) = f({\bf x}+t{\bf e}_P)$. 
Clearly $t \mapsto \tau_{t{\bf e}_P}$ defines a strongly continuous
semigroup with generator ${\bf e}_P\cdot {\bf p}$.

Notice that for $t$ sufficiently small, $\supp \tau_{t{\bf
    e}_P}u_{\gamma} \subset U_P(\epsilon/2)$.
Since $u_{\gamma} \in {\mathcal D}(\widetilde{{\mathbf H}}) \subset
L^{2}({\mathbb R}^{3|Q|}) \otimes H^{1}({\mathbb R}^{3|P|})$ by the 
induction hypothesis, we know that  
$$
  \partial_{{\bf e}_P} u_{\gamma} \in L^2(\R^{3N})\,,
$$
so $u_{\gamma} \in {\mathcal D}({\bf e}_P\cdot {\bf p})$ and
\begin{align}\label{diff-quotient}
  \lim_{t\to0}\frac{1}{t}(\tau_{t{\bf e}_P}u_{\gamma}-u_{\gamma})
  =\partial_{{\bf e}_P} u_{\gamma}\,,
\end{align}
in \(L^2(\R^{3N})\).

Let $v \in {\mathcal
D}({\mathbf H})$ and consider $\langle {\mathbf H}v, \partial_{{\bf
  e}_P}u_{\gamma} \rangle$. Using \eqref{diff-quotient} and
\eqref{eq:formsEqual}, we get 
\begin{align}\label{eq:something}
  \langle {\mathbf H}v, \partial_{{\bf e}_P} u_{\gamma} \rangle &=
  \lim_{t\rightarrow 0}t^{-1}\langle {\mathbf H} v,
  \tau_{t{\bf e}_P}u_{\gamma}-u_{\gamma}
  \rangle = \lim_{t\rightarrow 0} t^{-1} \mathfrak{q}(v,\tau_{t{\bf
      e}_P}u_{\gamma}-u_{\gamma}) 
  \nonumber\\ 
  &=
  \lim_{t\rightarrow 0} t^{-1}
  \widetilde{\mathfrak{q}}(v,\tau_{t{\bf e}_P}u_{\gamma}-u_{\gamma})\,. 
\end{align}
Since the translation \(\tau_{t{\bf e}_P}\) commutes with
$\widehat{{\mathbf H}}={\mathbf H}_P + {\mathbf H}_Q$ (see
\eqref{def:H-Q} and \eqref{def:H-P}), we get that, with $I_P$ from
\eqref{eq:IP}, 
\begin{align*}
  \widetilde{{\mathbf H}} \tau_{t{\bf e}_P} = \tau_{t{\bf e}_P} 
  \widetilde{{\mathbf H}} + [I_P,\tau_{t{\bf e}_P}]\,.
\end{align*}
Thus, using \eqref{eq:formsEqual}
\begin{align}\label{eq:anything}\nonumber
  \langle {\mathbf H}v, \partial_{{\bf e}_P} u_{\gamma} \rangle &=
  \lim_{t\rightarrow 0} t^{-1}
  \widetilde{\mathfrak{q}}(v,\tau_{t{\bf e}_P}u_{\gamma}-u_{\gamma}) 
  \\&
  =\lim_{t\to0} t^{-1}\langle v,
  \tau_{t{\bf e}_P}(\widetilde{{\mathbf H}}u_{\gamma})
  -\widetilde{{\mathbf H}}u_{\gamma}\rangle - 
  \langle v ,(\partial_{{\bf e}_P} I_P) u_{\gamma}\rangle\,. 
\end{align}
To prove that \(\partial_{{\bf e}_P} u_{\gamma}\in\mathcal{D}({\bf
  H}^*)=\mathcal{D}({\bf H})\) from this, it remains to 
show that when applying \(\partial_{{\bf e}_P}\) to \(\widetilde{{\mathbf H}}
  u_{\gamma}\) we obtain a function belonging to \(L^2(\R^{3N})\).
Then, from \eqref{eq:anything}, also 
$$
  {\bf H}(\partial_{{\bf e}_P} u_{\gamma})
  ={\bf H^*}(\partial_{{\bf e}_P} u_{\gamma})
  =\partial_{{\bf e}_P}(\widetilde{{\mathbf H}}
  u_{\gamma})-(\partial_{{\bf e}_P} I_P) u_{\gamma}\,. 
$$

By \eqref{eq:DiffUnder}, localization and (1) from the induction
hypothesis, we find 
\begin{align}\label{eq:split1}
  \widetilde{{\mathbf H}} u_{\gamma} = {\mathbf H} u_{\gamma} &=
  E \partial_{x_P}^{\gamma} (\varphi_{P,\epsilon} \psi) +
  \partial_{x_P}^{\gamma}  [{\mathbf T},\varphi_{P,\epsilon}] \psi
  \\\nonumber 
  &\quad
  - \sum_{j=1}^k 
  \partial_{x_P}^{\gamma_1+\cdots+\gamma_{j-1}}
  \big\{(\partial_{x_P}^{\gamma_j} I_P) 
  \partial_{x_P}^{\gamma_{j+1}+ \cdots + \gamma_k}
  (\varphi_{P,\epsilon} \psi)\big\}\,.
\end{align}
We will show that when applying \(\partial_{{\bf e}_P}\) to each term
on the right side of \eqref{eq:split1} 
we obtain a function belonging to \(L^2(\R^{3N})\). 

For the first term, since $\partial_{x_P}^{\gamma}
(\varphi_{P,\epsilon} \psi) \in {\mathcal D}(\widetilde{{\mathbf H}})
\subset L^{2}({\mathbb R}^{3|Q|}) \otimes H^{1}({\mathbb R}^{3|P|})$
by \eqref{eq:DomHat} and the induction hypothesis, we know that
\begin{align}
\label{eq:EtLed}
  \partial_{{\bf e}_P} \partial_{x_P}^{\gamma} (\varphi_{P,\epsilon}
  \psi) \in L^2({\mathbb R}^{3N})\,. 
\end{align}

For the third term, the function $I_P$ from \eqref{eq:IP}
satisfies $I_P \in {\mathcal B}^{\infty}({\mathbb R}^{3N})$, and, as
just shown, $\partial_{x_P}^{\alpha} (\varphi_{P,\epsilon} \psi) \in
L^2({\mathbb R}^{3N})$ for all $|\alpha|\leq k+1$, so, by Leibniz'
rule, 
\begin{align}\label{eq:ToLed}
  \partial_{{\bf e}_P} \Big(\sum_{j=1}^k 
  \partial_{x_P}^{\gamma_1+\cdots+\gamma_{j-1}}
  \big\{(\partial_{x_P}^{\gamma_j} I_P) 
  \partial_{x_P}^{\gamma_{j+1}+ \cdots + \gamma_k}
  (\varphi_{P,\epsilon} \psi)\big\}\Big)\in L^2({\mathbb R}^{3N})\,. 
\end{align}

Finally, we consider the commutator term $\partial_{x_P}^{\gamma}
[{\mathbf T},\varphi_{P,\epsilon}] \psi$ in \eqref{eq:split1}. 
Define $\varphi_1=\varphi_{P,\epsilon/4}$, $\varphi_2 = 1-\varphi_1$.
Notice that, by the definition of $f_1,f_2$,
\begin{align}\label{eq:SuppCond}
  f_1(8t/\epsilon) f_2(2t/\epsilon)=0\,.
\end{align}
By using that $f_1 + f_2 = 1$ we find 
\begin{align}
  \varphi_2 = \sum_{(\{s_j\},\{s_{j,k}\})} \prod_{j\in P}
  f_{s_j}(8|x_j|/\epsilon) \prod_{j\in P, k\in Q}
  f_{s_{j,k}}(8|x_j-x_k|/\epsilon)\,, 
\end{align}
where the sum is over all tuples $(\{s_j\},\{s_{j,k}\}) \in
\{1,2\}^{|P|+|P|\cdot|Q|}$ with at least one entry different from
$2$. 
Write  the commutator term $\partial_{x_P}^{\gamma}
[{\mathbf T},\varphi_{P,\epsilon}] \psi$ as 
\begin{align}\label{eq:furtherLoc}
  \partial_{x_P}^{\gamma}  [{\mathbf T},\varphi_{P,\epsilon}] \psi =
  \partial_{x_P}^{\gamma}  [{\mathbf T},\varphi_{P,\epsilon}] (\varphi_1\psi) + 
  \partial_{x_P}^{\gamma}  [{\mathbf T},\varphi_{P,\epsilon}] (\varphi_2\psi)\,,
\end{align}
The term with $\varphi_1$ we write, using Leibniz' rule, as
\begin{align*}
  \partial_{x_P}^{\gamma}  [{\mathbf T},\varphi_{P,\epsilon}] (\varphi_1\psi)
  =
  \sum_{\beta\leq \gamma}  \begin{pmatrix} \gamma \\ \beta \end{pmatrix}
  [{\mathbf T},\partial_{x_P}^{\beta} \varphi_{P,\epsilon}] 
  \partial_{x_P}^{\gamma-\beta} (\varphi_1\psi)\,,
\end{align*}
so,
\begin{align*}
  \partial_{{\bf e}_P} \partial_{x_P}^{\gamma}  [{\mathbf
    T},\varphi_{P,\epsilon}] (\varphi_1\psi) 
  &=
  \sum_{\beta\leq \gamma}  \begin{pmatrix} \gamma \\ \beta \end{pmatrix}
  [{\mathbf T},\partial_{{\bf e}_P} \partial_{x_P}^{\beta} \varphi_{P,\epsilon}] 
  \partial_{x_P}^{\gamma-\beta} (\varphi_1\psi)\\
  &\quad +
  \sum_{\beta\leq \gamma} 
  \begin{pmatrix} \gamma \\ \beta \end{pmatrix}
  [{\mathbf T},\partial_{x_P}^{\beta} \varphi_{P,\epsilon}] 
  \partial_{{\bf e}_P}\partial_{x_P}^{\gamma-\beta} (\varphi_1\psi)\,.
\end{align*}
By Lemma~\ref{lem:PSDO} and the induction hypothesis, we therefore see
that 
\begin{align}\label{eq:TreLed2}
  \partial_{{\bf e}_P} \partial_{x_P}^{\gamma}  [{\mathbf T},\varphi_{P,\epsilon}]
  (\varphi_1\psi) \in L^2({\mathbb R}^{3N})\,. 
\end{align}

Now we consider the term with $\varphi_2$ in \eqref{eq:furtherLoc}. 
We will prove that also
\begin{align}\label{eq:TreLed1bis}
  \partial_{{\bf e}_P} \partial_{x_P}^{\gamma} & \big[{\bf
    T},\varphi_{P,\epsilon}\big] (\varphi_2\psi) 
  \in L^2({\mathbb R}^{3N})\,. 
\end{align}
Since ${\bf T}$ is a finite sum and
\(\varphi_{P,\epsilon}\varphi_2=0\) it suffices, up to renumbering of
the terms, to prove that 
\begin{align}\label{eq:TreLed1}
  -\partial_{{\bf e}_P} \partial_{x_P}^{\gamma} &
  \big[\sqrt{p_1^2+m},\varphi_{P,\epsilon}\big] 
  (\varphi_2\psi) \nonumber\\
  &= \partial_{{\bf e}_P} \partial_{x_P}^{\gamma}
  \varphi_{P,\epsilon} \sqrt{p_1^2+m}\, (\varphi_2\psi) 
  \in L^2({\mathbb R}^{3N})\,.
\end{align}
\begin{proof}[Proof of \eqref{eq:TreLed1}]~\\
{\bf Case 1.} $1\in P$.\\
The case $P=\{1\}$ being immediate by Lemma~\ref{lem:PSDO}, we will
assume that $P_1 \neq \emptyset$, where $P_1 :=P \setminus \{1\}$. 

Since $\sqrt{p_1^2+m}$ commutes with multiplication operators in other
variables, and using the support condition \eqref{eq:SuppCond}, we
find 
\begin{align*}
  \varphi_{P,\epsilon} \sqrt{p_1^2+m}\,\varphi_2
  &=
  \prod_{j\in P_1} f_2(2|x_j|/\epsilon) \prod_{j\in P_1, k\in Q}
  f_2(2|x_j-x_k|/\epsilon) \nonumber\\ 
  &\times\Big\{ f_2(2|x_1|/\epsilon) \prod_{k\in Q}
  f_2(2|x_1-x_k|/\epsilon)\sqrt{p_1^2+m} \,f \Big\}\,, 
\end{align*}
with
\begin{align}
  f:=\sum_{\sigma} f_{\sigma_1}(8|x_1|/\epsilon) \prod_{k\in Q}
  f_{\sigma_k}(8|x_1-x_k|/\epsilon)\,, 
\end{align}
where the sum is over all $\sigma \in \{1,2\}^{1 + |Q|}$ with $\sigma
\neq (2,\ldots,2)$. Since at least one factor for each summand has to
be $f_1$ we find 
$$
  \supp f \subset \big\{ {\bf x} \,\big |\, \min\big( |x_1|,
  \min_{k\in Q} |x_1-x_k|\big) \leq \epsilon/4\big\}\,. 
$$
Thus, by the triangle inequality
$$
  \supp \Big(\prod_{j\in P_1} f_2(2|x_j|/\epsilon) \prod_{j\in P_1,
    k\in Q} f_2(2|x_j-x_k|/\epsilon) f\Big)\subset
  U_{P_1}(\epsilon/4)\,. 
$$
Since $\varphi_{P_1, \epsilon/4}=1$ on $U_{P_1}(\epsilon/4)$ we get the identity
\begin{align}\label{eq:ExtraLoc}
  \varphi_{P,\epsilon} \sqrt{p_1^2+m}\,\varphi_2
  =
  \big(\varphi_{P,\epsilon} \sqrt{p_1^2+m}\,\varphi_2\big) \varphi_{P_1,
    \epsilon/4}\,. 
\end{align}
By the induction hypothesis
\begin{align}\label{eq:GettingThere}
  \partial_{x_P}^{\gamma'}\big(\varphi_{P_1, \epsilon/4} \psi\big)
  \in L^2({\mathbb R}^{3N})\,,
\end{align}
for all $|\gamma'|\leq n$. Furthermore, since \(\supp\varphi_{P,
  \epsilon}\cap\,\supp\varphi_2=\emptyset\),  
Lemma~\ref{lem:PSDO} yields that 
$$
  (\partial_{x_P}^{\alpha}\varphi_{P,\epsilon})
  \sqrt{p_1^2+m}\,(\partial_{x_P}^{\beta}\varphi_2) (1+p_1^2)^M 
$$
is a bounded operator on $L^2({\mathbb R}^{3N})$ for all $\alpha,
\beta, M$. 

By Leibniz rule and \eqref{eq:ExtraLoc},
\begin{align}\label{eq:PatchLeibniz}
  &\partial_{x_P}^{\gamma'} \big( \varphi_{P,\epsilon}
  \sqrt{p_1^2+m}\,\varphi_2 \big)\nonumber\\ 
  &=
  \sum_{\alpha_1+\alpha_2+\alpha_3 = \gamma'}
  c_{\alpha_1,\alpha_2,\alpha_3} \Big\{
  \partial_{x_P}^{\alpha_1}\varphi_{P,\epsilon})
  \sqrt{p_1^2+m}\,(\partial_{x_P}^{\alpha_2}\varphi_2)
  (1+p_1^2)^M\Big\}\nonumber\\ 
  &\quad\quad\quad\quad\quad\quad\quad\quad\quad\quad\quad\quad 
  \times \big\{ (1+p_1^2)^{-M} \partial_{x_P}^{\alpha_3}\big(
  \varphi_{P_1, \epsilon/4} \psi\big)\big\}\,, 
\end{align}
for some constants $c_{\alpha_1,\alpha_2,\alpha_3}$.

By definition, $\partial_{x_P}^{\alpha}=\sum_{\beta\leq \alpha}
c_{\alpha,\beta} \partial_1^{\beta}\partial_{x_{P_1}}^{\alpha-\beta}$
for some constants $c_{\alpha,\beta}$. 
So using \eqref{eq:GettingThere} and choosing
$\partial_{x_P}^{\gamma'} = \partial_{{\bf e}_P}
\partial_{x_P}^{\gamma}$ and $M\geq |\gamma|+1$ in
\eqref{eq:PatchLeibniz}, we see that 
\begin{align*}
  \partial_{{\bf e}_P} \partial_{x_P}^{\gamma}\big(
  \varphi_{P,\epsilon} \sqrt{p_1^2+m}\,\varphi_2 \big) 
  \in L^2({\mathbb R}^{3N})\,.
\end{align*}
This finishes the proof of \eqref{eq:TreLed1} in the case $1 \in P$.

\noindent{\bf Case 2.} $1\notin P$.\\
This case is similar but simpler than Case 1. In this case we define
$P_1 = P$. Arguing as previously we realize that the identity
\eqref{eq:ExtraLoc} remains valid. Also \eqref{eq:GettingThere}
follows from the induction hypothesis. Since $P=P_1$, we can in this
case choose $M=0$ in \eqref{eq:PatchLeibniz} and get the desired
result. This finishes the proof of \eqref{eq:TreLed1} in the case
$1\notin P$ and combining with Case 1, we get the general result. 
\end{proof}

Combining \eqref{eq:EtLed}, \eqref{eq:ToLed}, \eqref{eq:TreLed1bis}, and
\eqref{eq:TreLed2}, we get that 
\begin{align}\label{eq:derOK}
  \partial_{{\bf e}_P}(\widetilde{{\mathbf H}} u_{\gamma}) \in
  L^2({\mathbb R}^{3N})\,. 
\end{align}
So we see from
\eqref{eq:anything} that for all \(v\in{\mathcal D}({\bf H})\),  
\begin{align}\label{eq:InDomain}
  \langle {\mathbf H}v, \partial_{{\bf e}_P} u_{\gamma} \rangle=
  \langle v, \partial_{{\bf e}_P}(\widetilde{{\mathbf H}} u_{\gamma}) -
  (\partial_{{\bf e}_P} I_P) u_{\gamma}\rangle\,.
\end{align}
From \eqref{eq:derOK}, \eqref{eq:InDomain}, and \eqref{eq:formsEqual} we
conclude that 
\begin{align}\label{eq:Conc1}
  \partial_{{\bf e}_P} u_{\gamma} \in {\mathcal D}({\mathbf
    H^*})\cap{\mathcal D}({\widetilde{{\mathbf H}}}^*)
  ={\mathcal D}({\mathbf H})
  \cap
  {\mathcal D}(\widetilde{{\mathbf H}})\,, 
\end{align}
and
\begin{align}\label{eq:Conc2}
  {\mathbf H}(\partial_{{\bf e}_P} u_{\gamma}) = \partial_{{\bf e}_P}(
  {\mathbf H} u_{\gamma}) - 
  (\partial_{{\bf e}_P} I_P) u_{\gamma}\,.
\end{align}
The equations \eqref{eq:Conc1} and \eqref{eq:Conc2} combine to give
the statement in Lemma~\ref{lem:S(k)} for \(k+1\). 

This finishes the induction step, and by
induction the statement in Lemma~\ref{lem:S(k)}
therefore holds for all \(k\in\mathbb{N}\cup\{0\}\).
\end{pf}
As mentioned above, this finishes the
proof of Proposition~\ref{lem:diff_parallel}. 
\end{proof}

\appendix

\section{Auxiliary results from operator theory}
\label{app:A}
In the proof of Lemma~\ref{lem:diff_parallel} we need the following
consequence of the Davies-Faris Theorem (\cite[Theorem X.31]{R&S2}). 
\begin{lemma}
\label{lem:davies_faris}
Suppose $T\geq 0$ is self-adjoint with domain \(\mathcal{D}(T)\), and
that \(T\) is
the generator of a positivity 
preser\-ving semigroup. Let $V$ be a positive multiplication operator,
which is bounded relative to $T$. Then ${\mathbf H}= T+V$ is
self-adjoint on ${\mathcal D}(T)$.
\end{lemma}
\begin{proof}
Choose $g>0$ such that $gV$ is relatively bounded with
respect to $T$ with bound $a < 1$. We will prove by
induction that $K_n= T+ ngV$ is self-adjoint on ${\mathcal
D}(T)$ for all $n \in {\mathbb N}$. In order to do so, let
us consider the following statement
$S(n)$:\\
\begin{enumerate}
\item $K_n= T+ ngV$ is self-adjoint on ${\mathcal
D}(T)$.
\item
$\| gV \varphi \| \leq a \| (K_n+1) \varphi \| \text{ for all }
\varphi \in {\mathcal D}(T).$
\item $K_n$ is the generator of a positivity preserving
semigroup.
\end{enumerate}
\ %
\\
Note first that $S(0)$ is true by assumption.

Suppose now $S(n)$ holds true for some $n \geq 0$. 
By $S(n)$ point (2), $gV$ is a small operator
perturbation of $K_n$, so $K_{n+1}=K_n+gV$ is (by the
Kato-Rellich Theorem \cite[Theorem X.12]{R&S2}) self-adjoint on
${\mathcal D}(K_n)={\mathcal D}(T)$. Furthermore, using the Trotter
product formula \cite[Theorem X.51]{R&S2} and the induction
hypothesis, it is easy to see that $e^{-t K_{n+1}}$ is positivity
preserving (for \(t>0\)). 
Then, by the Davies-Faris Theorem
\cite[Theorem X.31]{R&S2}, it follows that $gV$
satisfies the bound
$$
  \| gV \varphi \| \leq a \| (K_{n+1}+1) \varphi \| \text{ for all }
  \varphi \in {\mathcal D}(T)\,.
$$
Therefore \(S(n+1)\) holds.
This finishes the proof that \(S(n)\) implies \(S(n+1)\) for any
\(n\ge0\). 

The proof of Lemma~\ref{lem:davies_faris} now follows by induction. 
\end{proof}

We also state the following lemma which is used repeatedly in
Section~\ref{parallel}. The proof is standard and is omitted.
\begin{lemma}
\label{lem:PSDO}
Let $\chi, \phi \in {\mathcal B}^{\infty}({\mathbb R}^{3N})$ have
disjoint support and let $m\geq 0$. Then $[\sqrt{p_j^2+m}, \varphi ]$
defines a bounded operator on $H^s({\mathbb R}^{3N})$ for all $s \in
{\mathbb R}$ and $(1+p_j^2)^M \chi [\sqrt{p_j^2+m}, \varphi
](1+p_j^2)^M$ is a bounded operator on $L^2({\mathbb R}^{3N})$ for all
$M$. 
\end{lemma}

\begin{acknowledgement}
Parts of this work have been carried out at various
institutions, whose hospitality is gratefully acknowledged:
Mathematisches Forschungsinstitut Oberwolfach 
(SF, T\O S), Erwin Schr\"{o}\-dinger Institute (SF, T\O
S), Universit\'{e} Paris-Sud (T\O S), and the IH\'ES (T\O S).
SF is partially supported by a Skou Grant and a Young Eliteresearcher award from the 
Danish councils for independent research, a grant from the Lundbeck Foundation, and the 
European Research Council under the European Community's Seventh Framework Programme 
(FP7/2007-2013)/ERC grant agreement n$^{\circ}$ 202859. T{\O}S is
partially supported by The 
Danish Natural Science Research
Council, under the grant `Mathematical Physics and Partial Differential
Equations'.
\end{acknowledgement}

\def\cprime{$'$}
\providecommand{\bysame}{\leavevmode\hbox to3em{\hrulefill}\thinspace}
\providecommand{\MR}{\relax\ifhmode\unskip\space\fi MR }
\providecommand{\MRhref}[2]{%
  \href{http://www.ams.org/mathscinet-getitem?mr=#1}{#2}
}
\providecommand{\href}[2]{#2}

\end{document}